\documentclass[11pt,letterpaper]{article}

\usepackage{amsmath,amsfonts,amsthm,amssymb,stmaryrd,relsize}
\theoremstyle{plain}

\numberwithin{equation}{section}

\newtheorem{thm}{Theorem}[section]
\newtheorem{lem}[thm]{Lemma}
\newtheorem{cor}[thm]{Corollary}

\newenvironment{exam}
{\begin{flushleft}\textbf{Example}.\enspace}%
{\end{flushleft}}

\allowdisplaybreaks  

\newcommand{\complex}{{\mathbb C}}
\newcommand{\real}{{\mathbb R}}

\newcommand{\tbullet}{\raise .4ex\hbox{\tiny$\bullet$}} 

\newcommand{\phihat}{\widehat{\phi}}
\newcommand{\pscripthat}{\widehat{\pscript}}
\newcommand{\rhotilde}{\widetilde{\rho}}

\newcommand{\rmtr}{\mathrm{tr\,}}
\newcommand{\ityes}{\textit{yes}}
\newcommand{\itno}{\textit{no}}

\newcommand{\ascript}{\mathcal{A}}
\newcommand{\cscript}{\mathcal{C}}
\newcommand{\escript}{\mathcal{E}}
\newcommand{\iscript}{\mathcal{I}}
\newcommand{\lscript}{\mathcal{L}}
\newcommand{\mscript}{\mathcal{M}}
\newcommand{\oscript}{\mathcal{O}}
\newcommand{\pscript}{\mathcal{P}}
\newcommand{\sscript}{\mathcal{S}}

\newcommand{\ab}[1]{\left|#1\right|}
\newcommand{\doubleab}[1]{\left|\left|#1\right|\right|}
\newcommand{\brac}[1]{\left\{#1\right\}}
\newcommand{\paren}[1]{\left(#1\right)}
\newcommand{\sqbrac}[1]{\left[#1\right]}
\newcommand{\elbows}[1]{{\left\langle#1\right\rangle}}
\newcommand{\ket}[1]{{\left|#1\right>}}
\newcommand{\bra}[1]{{\left<#1\right|}}

\begin{document}

\title{QUANTUM INSTRUMENTS\\ AND CONDITIONED OBSERVABLES}
\author{Stan Gudder\\ Department of Mathematics\\
University of Denver\\ Denver, Colorado 80208\\
sgudder@du.edu}
\date{}
\maketitle

\begin{abstract}
Observables and instruments have played significant roles in recent studies on the foundations of quantum mechanics. Sequential products of effects and conditioned observables have also been introduced. After an introduction in Section~1, we review these concepts in Section~2. Moreover, it is shown how these ideas can be unified within the framework of measurement models. In Section~3, we illustrate these concepts and their relationships for the simple example of a qubit Hilbert space. Conditioned observables and their distributions are studied in Section~4. Section~5 considers joint probabilities of observables. We introduce a definition for joint probabilities and discuss why we consider this to be superior to the standard definition.
\end{abstract}

\section{Introduction}  
This article is a continuation of the author's work on conditioned observables in quantum mechanics \cite{gud20}. For the reader's convenience, we first review the concepts needed in the present paper. We shall only consider quantum systems described by finite-dimensional Hilbert spaces. Although this is a strong restriction, it is general enough to include the important subjects of quantum computation and information theory \cite{hz12,nc00}.

In Section~2, we review the definitions of quantum effects, observables and instruments \cite{bgl95,hz12,kra83,nc00}. We consider the sequential product and conditioning of effects and observables \cite{gg04,gg02,gn01,gud20}. Quantum operations, channels and instruments are discussed. The idea of different instruments measuring an observable is presented and the special role of the L\"uders instrument is emphasized. We also discuss the unifying framework of measurement models \cite{bcl95,bgl95,hz12}.

The various concepts presented in Section~2 are illustrated for the simplest case of a qubit Hilbert space in Section~3. In particular, we discuss spin component observables. Section~4 studies conditioned observables. Complementary observables and their relationship to mutually unbiased bases are presented. We also consider observable probability distributions. Finally, in Section~5 we introduce what we consider to be the natural and correct definition of joint probabilities of observables. Moreover, we discuss why we believe this to be superior to the standard definition.

\section{Effects, Observables and Instruments}  
Let $\lscript (H)$ be the set of linear operators on a finite-dimensional complex Hilbert space $H$. For $S,T\in\lscript (H)$ we write $S\le T$ if $\elbows{\phi ,S\phi}=\elbows{\phi ,T\phi}$ for all $\phi\in H$. We define the set of \textit{effects} by
\begin{equation*}
\escript (H)=\brac{a\in\lscript (H)\colon 0\le a\le I}
\end{equation*}
where $0$, $I$ are the zero and identity operators, respectively. The effects correspond to \ityes -\itno\ experiments and
$a\in\escript (H)$ is said to \textit{occur} when a measurement of $a$ results in the value \ityes . We denote the set of projections on $H$ by $\pscript (H)$. It is clear that $\pscript (H)\subseteq\escript (H)$ and we call the elements of
$\pscript (H)$ \textit{sharp effects} \cite{hz12,kra83,nc00}. A one-dimensional projection $P_\phi =\ket{\phi}\bra{\phi}$, where
$\doubleab{\phi}=1$, is called an \textit{atom}. If $\phi\in H$ with $\phi\ne 0$, we write $\phihat =\phi\big/\doubleab{\phi}$. We then have
\begin{equation*}
P_{\phihat}=\tfrac{1}{\doubleab{\phi}^2}\,\ket{\phi}\bra{\phi}
\end{equation*}

We call $\rho\in\escript (H)$ a \textit{partial state} if $\rmtr (\rho )\le 1$ and $\rho$ is a \textit{state} if $\rmtr (\rho )=1$. We denote the set of states by $\sscript (H)$ and the set of partial states by $\sscript _p(H)$. If $\rho\in\sscript (H)$,
$a\in\escript (H)$ we call $\pscript _\rho (a)=\rmtr(\rho a)$ the \textit{probability that} $a$ \textit{occurs} in the state $\rho$. Of course, $0\le\pscript _\rho (a)\le 1$. If $a,b\in\escript (H)$ and $a+b\le I$ we write $a\perp b$. When $a\perp b$ we have that $a+b\in\escript (H)$ and $\pscript _\rho (a+b)=\pscript _\rho (a)+\pscript _\rho (b)$. If $P_\phi$ is an atom, then we call
$P_\phi$ (and $\phi$) a \textit{pure state}. We then write
\begin{equation*}
\pscript _\phi(a)=\pscript _{P_\phi}(a)=\rmtr (P_\phi a)=\elbows{\phi ,a\phi}
\end{equation*}
If $\phi$ and $\psi$ are pure states, we call $\ab{\elbows{\phi ,\psi}}^2$ the \textit{transition probability} from $\phi$ to $\psi$.

We denote the unique positive square-root of $a\in\escript (H)$ by $a^{1/2}$. For $a,b\in\escript (H)$, their
\textit{sequential product} is the effect $a\circ b=a^{1/2}ba^{1/2}$ where $a^{1/2}ba^{1/2}$ is the usual operator product \cite{gg04,gg02,gn01,lud51}. We interpret $a\circ b$ as the effect that results from first measuring $a$ and then measuring
$b$. It can be shown that $a\circ b\le a$ and that $a\circ b=b\circ a$ if and only if $ab=ba$. If $ab=ba$, we say that $a$ and
$b$ are \textit{compatible} and interpret this physically as meaning that $a$ and $b$ do not interfere. We also call $a\circ b$ the effect $b$ \textit{conditioned on} the effect $a$ and write $(b\mid a)=a\circ b$. Notice that if $b_1,b_2\in\escript (H)$ with $b_1\perp b_2$, then $(b_1+b_2\mid a)=(b_1\mid a)+(b_2\mid a)$. Moreover, $\escript (H)$ is convex and if 
$\lambda _i\ge 0$ with $\sum\lambda _i=1$, then 
\begin{equation*}
\paren{\sum\lambda _ib_i\mid a}=\sum\lambda _i(b_i\mid a)
\end{equation*}
so $b\mapsto (b\mid a)$ is an affine function. Of course, $a\mapsto (b\mid a)$ is not an affine function in general.

If $\rho\in\sscript (H)$ and $a\in\escript (H)$ with $\rho\circ a\ne 0$, since $\rho\circ a\le \rho$ we have that
\begin{equation*}
\rmtr\sqbrac{(\rho\mid a)}=\rmtr (a\circ\rho )=\rmtr (\rho\circ a)\le\rmtr (\rho )=1
\end{equation*}
Hence, $(\rho\mid a)\in\sscript _p(H)$ and for $a\in\escript (H)$ we obtain
\begin{align*}
\pscript _\rho\sqbrac{(b\mid a)}&=\rmtr\sqbrac{\rho (b\mid a)}=\rmtr(\rho a\circ b)=\rmtr\sqbrac{(a\circ\rho )b}\\
  &=\rmtr\sqbrac{(\rho\mid a)b}
\end{align*}
If $\pscript _\rho (a)=\rmtr (\rho a)\ne 0$, we can form the state $(\rho\mid a)/\rmtr (\rho a)$. Then as a function of $b$
\begin{equation}                
\label{eq21}
\pscripthat _\rho\sqbrac{(b\mid a)}=\frac{\pscript _\rho\sqbrac{(b\mid a)}}{\pscript _\rho (a)}
\end{equation}
becomes a probability measure on $\escript (H)$ and we call $\pscripthat _\rho\sqbrac{(b\mid a)}$ the
\textit{conditional probability of} $b$ \textit{given} $a$.

For a finite set $\Omega _A$, an \textit{observable with value-space} $\Omega _A$ is a subset
$A=\brac{A_x\colon x\in\Omega _A}$ of $\escript (H)$ such that $\sum\limits _{x\in\Omega _A}A_x=I$. We interpret $A_x$ as the effect that occurs when $A$ has the value $x$. The condition $\sum A_x=I$ ensures that $A$ has one of the values
$x\in\Omega _A$ when $A$ is measured. Defining $A_X=\sum\limits _{x\in X}A_x$ for all $X\subseteq\Omega _A$, we see that $X\mapsto A_X$ is a finite positive operator-valued measure on $H$ \cite{bcl95, bgl95,hz12,kra83,nc00}. If
$A_x\in\pscript (H)$ for all $x\in\Omega _A$, we call $A$ a \textit{sharp observable}. The effects $A_x$ for a sharp observable commute and are mutually orthogonal \cite{hz12,nc00}. If the effects $A_x$ are atoms, we say that $A$ is \textit{atomic}. In this case, $A_x=P_{\phi _x}$ where $\brac{\phi _x\colon x\in\Omega _A}$ is an orthonormal basis for $H$. We denote the set of observables on $H$ by $\oscript (H)$.

For $A,B\in\oscript (H)$ with $A=\brac{A_x\colon x\in\Omega _A}$ and $B=\brac{B_y\colon y\in\Omega _B}$ we define their \textit{sequential product} $A\circ B$ \cite{bcl95} to be the observables with value-space $\Omega _A\times\Omega _B$ and
\begin{equation*}
A\circ B=\brac{A_x\circ B_y\colon (x,y)\in\Omega _A\times\Omega _B}
\end{equation*}
The observable $B$ \textit{conditioned by} the observable $A$ \cite{gud20} has value-space $\Omega _B$ and is defined by
\begin{equation*}
(B\mid A)=\brac{\sum _{x\in\Omega _A}A_x\circ B_y\colon y\in\Omega _B}
  =\brac{\sum _{x\in\Omega _A}(B_y\mid A_x)\colon y\in\Omega _B}
\end{equation*}
We denote the effects in $A\circ B$ and $(B\mid A)$ by $(A\circ B)_{(x,y)}=A_x\circ B_y$ and
$(B\mid A)_y=\sum\limits _{x\in\Omega _A}A_x\circ B_y$, respectively. We say that $A$ and $B$ \textit{commute} if $a_xb_y=b_ya_x$ for all $x\in\Omega _A$, $y\in\Omega _B$. If $A$ and $B$ commute, then $(B\mid A)=B$. We do not know whether the converse holds. However, we have the following result.

\begin{lem}    
\label{lem21}
If $(B\mid A)=B$ and $A$ is sharp, then $A$ and $B$ commute.
\end{lem}
\begin{proof}
Since $(B\mid A)=B$ and $A$ is sharp, we have that
\begin{equation*}
\sum _{x\in\Omega _A}A_xB_yA_x=\sum _{x\in\Omega _A}A_x\circ B_y=B_y
\end{equation*}
for every $y\in\Omega _B$. Since the $A_x$'s are mutually orthogonal we obtain
\begin{equation*}
A_xB_yA_x=A_xB_y=B_yA_x
\end{equation*}
For all $x\in\Omega _A$, $y\in\Omega _B$. Hence, $A$ and $B$ commute.
\end{proof}
If $\rho\in\sscript (H)$ and $A\in\oscript (H)$, we define the state $\rho$ \textit{conditioned on} $A$ \cite{gud20} by
\begin{equation}                
\label{eq22}
(\rho\mid A)=\sum _{x\in\Omega _A}(\rho\mid A_x)=\sum _{x\in\Omega _A}A_x^{1/2}\rho A_x^{1/2}
\end{equation}

An \textit{operation} on $H$ is a completely positive affine map $\ascript\colon\sscript _p(H)\to\sscript _p(H)$ \cite{hz12,nc00}. Thus, if $\lambda _i\ge 0$, $\sum\lambda _i=1$ and $\rho _i\in\sscript _p(H)$, $i=1,2,\ldots ,n$, then
\begin{equation*}
\ascript\paren{\sum _{i=1}^n\lambda _i\rho _i}=\sum _{i=1}^n\lambda _i\ascript (\rho _i)
\end{equation*}
We call an operation $\ascript$ a \textit{channel} if $\ascript (\rho )\in\sscript (H)$ for all $\rho\in\sscript (H)$. We denote the set of channels on $H$ by $\cscript (H)$. Notice that if $a\in\escript (H)$, then the map $\rho\mapsto (\rho\mid a)$ is an example of an operation and if $A\in\oscript (H)$ then $\rho\mapsto (\rho\mid A)$ is a channel. For a finite set $\Omega _\iscript$, an \textit{instrument with value-space} $\Omega _\iscript$ is a set of operations
$\iscript =\brac{\iscript _x\colon x\in\Omega _\iscript}$ such that $\sum _{x\in\Omega _\iscript}\iscript _x\in\cscript (H)$. Defining $\iscript _X$ for $X\subseteq\Omega _\iscript$ by $\iscript _X=\sum\limits _{x\in X}\iscript _x$ we see that
$X\mapsto\iscript _X$ is an operation-valued measure on $H$ \cite{bgl95,hz12,kra83}. If $A\in\oscript (H)$, we say that an instrument $\iscript$ is $A$-\textit{compatible} if $\Omega _\iscript =\Omega _A$ and
\begin{equation*}
\pscript _\rho (A_X)=\rmtr\sqbrac{\iscript _X(\rho )}
\end{equation*}
for all $\rho\in\sscript (H)$, $X\subseteq\Omega _A$. To show that $\iscript$ is $A$-compatible, it is sufficient to show that
\begin{equation*}
\pscript _\rho (A_x)=\rmtr\sqbrac{\iscript _x(\rho )}
\end{equation*}
for all $\rho\in\sscript (H)$ and $x\in\Omega _A$.

We view an $A$-compatible instrument as an apparatus that can be employed to measure the observable $A$. If $\iscript$ is an instrument, then there is a unique $A^\iscript\in\oscript (H)$ such that $\iscript$ is $A^\iscript$-compatible \cite{hz12}. It is clear that $\iscript$ is $A$-compatible if and only if $A^\iscript =A$. On the other hand, if $A\in\oscript (H)$, then there are many $A$-compatible instruments. For example, if $\eta\in\sscript (H)$ then the \textit{trivial instrument}
$\iscript _X(\rho )=\rmtr (\rho A_x)\eta$ is $A$-compatible. In this work, an important $A$-compatible instrument is given by the \textit{L\"uders instrument} \cite{hz12} $\lscript _x^A(\rho )=A_x\circ\rho$. We then have
\begin{equation*}
\lscript _X^A(\rho )=\sum _{x\in X}A_x\circ\rho =\sum _{x\in X}A_x^{1/2}\rho A_x^{1/2}
\end{equation*}
Notice that $\lscript _{\Omega _A}^A(\rho )=(\rho\mid A)$ as in \eqref{eq22}. If
$A=\brac{\ket{\phi _x}\bra{\phi _x}\colon x\in\Omega _A}$ is atomic, we obtain
\begin{equation*}
\lscript _X^A(\rho )=\sum _{x\in X}\elbows{\phi _x,\rho\phi _x}\ket{\phi _x}\bra{\phi _x}
  =\sum _{x\in X}\pscript (P_{\phi _x})P_{\phi _x}
\end{equation*}

The duality between observables and instruments is emphasized by the unifying studies of measurement models \cite{bcl95,bgl95,hz12,oza85}. A \textit{measurement model} is a 5-tuple $\mscript =(H,K,\eta ,\nu ,F)$ where $H,K$ are Hilbert spaces called the \textit{base} and \textit{probe} systems, respectively, $\eta\in\sscript (K)$ is an \textit{initial state},
$\nu\colon\sscript (H\otimes K)\to\sscript (H\otimes K)$ is a channel describing the measurement interaction between the base and probe systems and $F\in\oscript (K)$ is the \textit{pointer observable}. The instrument on $H$ defined by
\begin{equation}                
\label{eq23}
\iscript _X^\mscript (\rho )=\rmtr _K\sqbrac{\nu (\rho\otimes\eta )(I\otimes F_X)} 
\end{equation}
is called the \textit{model instrument} where $X\subseteq\Omega _F=\Omega _\iscript$ and $\rmtr _K$ is the
\textit{partial trace} \cite{hz12,nc00}. The unique observable $B^\mscript\in\oscript (H)$ defined by
$B^\mscript =A^{\iscript ^\mscript}$ is the \textit{model observable}. We then have the \textit{probability reproducing condition}
\begin{equation}                
\label{eq24}
\rmtr (\rho B_X^\mscript )=\rmtr\sqbrac{\iscript _X^\mscript (\rho )}=\rmtr\sqbrac{\nu (\rho\otimes\eta )(I\otimes F_X)}
\end{equation}
for all $\rho\in\sscript (H)$, $X\subseteq\Omega _{B^\mscript}=\Omega _F$.

Thus, any measurement model $\mscript$ determines a unique instrument $\iscript ^\mscript$ and a unique observable
$B^\mscript$. Conversely, for any instrument $\iscript$ there exist many measurement models $\mscript$ such that
$\iscript =\iscript ^\mscript$ and for any observable $B$ there exist many model measurements $\mscript$ such that
$B=B^\mscript$. These are shown in the next two results.

\begin{thm}    
\label{thm22}
{\rm (Ozawa \cite{hz12,nc00})} For any instrument $\iscript$ on $H$ there exists a measurement model
$\mscript =(H,K,\eta ,\nu ,F)$ where $\eta$ is a pure state, $\nu$ is a unitary channel $\nu (\mu )=U\mu U^*$ and $F$ is a sharp observable such that $\iscript =\iscript ^\mscript$.
\end{thm}

\begin{cor}    
\label{cor23}
For any $B\in\oscript (H)$ there exists a measurement model $\mscript$ as in Theorem~\ref{thm22} such that
$B=B^\mscript$.
\end{cor}
\begin{proof}
Given $B\in\oscript (H)$ there exists a $B$-compatible instrument $\iscript$. By Theorem~\ref{thm22}, there exists a measurement model satisfying the given conditions such that $\iscript =\iscript ^\mscript$. Then
$B=A^\iscript =A^{\iscript ^\mscript}$ so $B=B^\mscript$.
\end{proof}

We now continue this study to include sequential products of observables and conditioned observables. If $A,B\in\oscript (H)$, then $A\circ B\in\oscript (H)$. By Corollary~\ref{cor23}, there exists a measurement model $\mscript =(H,K,\eta ,\nu ,F)$ satisfying the conditions of Theorem~\ref{thm22} such that $A\circ B=B^\mscript$. We then have that
\begin{equation*}
\Omega _F=\Omega _{A\circ B}=\Omega _A\times\Omega _B
\end{equation*}
and by \eqref{eq24}
\begin{equation}                
\label{eq25}
\rmtr\sqbrac{\rho (A\circ B)_Z}=\rmtr\sqbrac{U(\rho\otimes P_\phi )U^*(I\otimes F_Z)}
\end{equation}
for every $\rho\in\sscript (H)$ and $Z\subseteq\Omega _A\times\Omega _B$. In particular, for every $(x,y)\in\Omega _A\otimes\Omega _B$, $\rho\in\sscript (H)$ we obtain
\begin{equation}                
\label{eq26}
\rmtr (\rho A_x\circ B_y)=\rmtr\sqbrac{\rho (A\circ B)_{(x,y)}}=\rmtr\sqbrac{U(\rho\otimes P_\phi )U^*(I\otimes F_{(x,y)})}
\end{equation}
The advantage of \eqref{eq25} and \eqref{eq26} is that the statistics of $A\circ B$, which may be unsharp, is described by the sharp observable $F$ and as we have seen, sharp observables are simpler than general unsharp ones. In particular, the effects $F_{(x,y)}$ commute and are mutually orthogonal. Applying \eqref{eq26}, we conclude that
\begin{equation*}
\rmtr (\rho A_x)=\rmtr (\rho A_x\circ B_{\Omega _B})
  =\rmtr\sqbrac{U(\rho\otimes P_\phi )U^*\paren{I\otimes\sum _{y\in\Omega _B}F_{(x,y)}}}
\end{equation*}
so $A$ is described by the sharp observable
\begin{equation}                
\label{eq27}
\brac{\sum _{y\in\Omega _B}F_{(x,y)}\colon x\in\Omega _A}
\end{equation}
Considering $(B\mid A)$ we have by \eqref{eq26} that
\begin{align*}
\rmtr\sqbrac{\rho (B\mid A)_y}&=\rmtr\paren{\rho\sum _{x\in\Omega _A}A_x\circ B_y}
   =\sum _{x\in\Omega _A}\rmtr (\rho A_x\circ B_y)\\\noalign{\medskip}
   &=\rmtr\sqbrac{U(\rho\otimes P_\phi )U^*\paren{I\otimes\sum _{x\in\Omega _A}F_{(x,y)}}}
\end{align*}
so $(B\mid a)$ is described by the sharp observable
\begin{equation}                
\label{eq28}
\brac{\sum _{x\in\Omega _A}F_{(x,y)}\colon y\in\Omega _B}
\end{equation}
where \eqref{eq27} and \eqref{eq28} commute even though $A$ and $(B\mid A)$ need not. But this is taking us away from our primary mission so we leave a further study to later work.

\section{Qubit Hilbert Space}  
This section illustrates the concepts presented in Section~2 for the simplest case of a qubit Hilbert space $H=\complex ^2$ with the usual inner product. Let $\phi =(1,0)$, $\phi '=(0,1)$ be the standard orthonormal basis for $\complex ^2$ Relative to this bases, the \textit{Pauli operators} have the matrix forms
\begin{equation*}
\sigma _x=\begin{bmatrix}0&1\\1&0\\\end{bmatrix},\quad\sigma _y=\begin{bmatrix}0&-i\\i&0\\\end{bmatrix},\quad
  \sigma _z=\begin{bmatrix}1&0\\0&-1\\\end{bmatrix}
\end{equation*}
Letting $\sigma =(\sigma _x,\sigma _y,\sigma _z)$, every $\rho \in\sscript (H)$ has the form
\begin{equation*}
\rho =\tfrac{1}{2}\,(I+r\tbullet\sigma )
\end{equation*}
where $r\in\real ^3$ with $\doubleab{r}\le 1$ and $\tbullet$ is the usual dot product in $\real ^3$ \cite{hz12}. The eigenvalues of $\rho$ are $\lambda _\pm =\tfrac{1}{2}\,\paren{1\pm\doubleab{r}}$. We have that $\lambda _+=1$ and $\lambda _-=0$ if and only if $\doubleab{r}=1$ so $\rho$ is pure if and only if $\doubleab{r}=1$. Every $a\in\escript (H)$ has the form
\begin{equation*}
a=\tfrac{1}{2}\,(\alpha I+n\tbullet\sigma )
\end{equation*}
where $\doubleab{n}\le\alpha\le 2-\doubleab{n}$ and positivity is equivalent to $\doubleab{n}\le\alpha$ \cite{hz12}.

For $n\in\real ^3$ with $\doubleab{n}=1$, define the atoms $S_\pm ^n=\tfrac{1}{2}\,(I\pm n\tbullet\sigma )$. The atomic observable $S^n=\brac{S_+^n,S_-^n}$ is called the \textit{spin component observable in direction} $n$ \cite{hz12}. Then
$S_+^n$ is the effect for which the spin component is $+$ and $S_-^n$ is the effect for which the spin component is $-$ in the direction $n$ and the value-space $\Omega _{S^n}=\brac{+,-}$. The effect $S_+^n$ is the $1$-dimensional projection
\begin{equation*}
S_+^n=\tfrac{1}{2}\,\begin{bmatrix}1+n_3&n_1-in_2\\n_1+in_2&1-n_3\\\end{bmatrix}
\end{equation*}
and $S_-^n=I-S_+^n$. Suppose we measure $S^m$ first and $S^n$ second. Then
\begin{align}                
\label{eq31}
S^m\circ S^n&=\brac{S_+^m\circ S_+^n,S_+^m\circ S_-^n,S_-^m\circ S_+^n,S_-^m\circ S_-^n}\\
  \intertext{and}          
\label{eq32}       
(S^n\mid S^m)&=\brac{S_+^m\circ S_+^n+S_-^m\circ S_+^n,S_+^m\circ S_+^n+S_-^m\circ S_-^n}
\end{align}
The observables in \eqref{eq31} and \eqref{eq32} are not sharp even though $S^m$ and $S^n$ are sharp.

To illustrate, let $m=(0,0,1)$ and $n=(1,0,0)$. These correspond to spin measurements in the $z$ and $x$ directions, respectively. Then
\begin{equation*}
S_+^m=\begin{bmatrix}1&0\\0&0\\\end{bmatrix},\quad S_-^m=\begin{bmatrix}0&0\\0&1\\\end{bmatrix},\quad
  S_+^n=\tfrac{1}{2}\begin{bmatrix}1&1\\1&1\\\end{bmatrix},\quad
  S_-^n=\tfrac{1}{2}\begin{bmatrix}1&-1\\-1&1\\\end{bmatrix},
\end{equation*}
and we have that
\begin{align*}
S_+^m\circ S_+^n&=S_+^m\circ S_-^n=\tfrac{1}{2}\,S_+^m\\
S_-^m\circ S_+^n&=S_-^m\circ S_-^n=\tfrac{1}{2}\,S_-^m
\end{align*}
Hence
\begin{align*}
S^m\circ S^n&=\brac{\tfrac{1}{2}\,S_+^m,\tfrac{1}{2}\,S_+^m,\tfrac{1}{2}\,S_-^m,\tfrac{1}{2}\,S_-^m}\\
(S^n\mid S^m)&=\brac{(S^n\mid S^m)_+,(S^n\mid S^m)_-}=\brac{\tfrac{1}{2}\,I,\tfrac{1}{2}\,I}
\end{align*}
and similar formulas hold for $S^n\circ S^m$ and $(S^m\mid S^n)$. Notice that $(S^n\mid S^m)=(S^m\mid S^n)$ but $S^m$ and $S^n$ do not commute.

We now find the general form of $S^m\circ S^n$ and $(S^n\mid S^m)$. The normalized eigenvector of $S_+^n$ with corresponding eigenvalue $1$ is
\begin{equation*}
\phi _+^n=\frac{1}{\sqrt{2(1-n_3)\,}}\begin{bmatrix}n_1-in_2\\1-n_3\\\end{bmatrix}\hbox{if }n_3\ne 1,\quad
  \phi _+^n=\begin{bmatrix}1\\0\\\end{bmatrix}\hbox{if }n_3=1
\end{equation*}
and corresponding to eigenvalue $0$ we have
\begin{equation*}
\phi _-^n=\frac{1}{\sqrt{2(1+n_3)\,}}\begin{bmatrix}in_2-n_1\\1+n_3\\\end{bmatrix}\hbox{if }n_3\ne -1,\quad
  \phi _-^n=\begin{bmatrix}0\\1\\\end{bmatrix}\hbox{if }n_3=-1
\end{equation*}
Since $S_+^n=\ket{\phi _+^n}\bra{\phi _+^n}$ and $S_-^n=\ket{\phi _-^n}\bra{\phi _-^n}$, we obtain
\begin{align*}
S^m&\circ S^n\\
  &=\brac{\ab{\elbows{\phi _+^m,\phi _+^n}}^2S_+^m,\ab{\elbows{\phi _+^m,\phi _-^n}}^2S_+^m,
   \ab{\elbows{\phi _-^m,\phi _+^n}}^2S_-^m,\ab{\elbows{\phi _-^m,\phi _-^n}}^2S_-^m}
\end{align*}
Moreover, we have that
\begin{align*}
(S^n\mid S^m)_+&=\ab{\elbows{\phi _+^m,\phi _+^n}}^2S_+^m+\ab{\elbows{\phi _-^m,\phi _+^n}}^2S_-^m\\
(S^n\mid S^m)_-&=\ab{\elbows{\phi _+^m,\phi _-^n}}^2S_+^m+\ab{\elbows{\phi _-^m,\phi _-^n}}^2S_-^m
\end{align*}
Letting $a=\ab{\elbows{\phi _+^m,\phi _+^n}}^2$ we conclude that
\begin{align*}
(S^n\mid S^m)_+&=(2a-1)S_+^m+(1-a)I\\
(S^n\mid S^m)_-&=(1-2a)S_+^m+aI
\end{align*}

For another example, the L\"uders channel \cite{lud51} for $S^n$ becomes
\begin{align*}
\lscript ^{S^n}(\rho )&=S_+^n\circ\rho +S_-^n\circ\rho =S_+^n\rho S_+^n+S_-^n\rho S_-^n\\
  &=\elbows{\phi _+^n,\rho\phi _+^n}P_{\phi _+^n}+\elbows{\phi _-^n,\rho\phi _-^n}P_{\phi _-^n}=(\rho\mid A)
\end{align*}
as in \eqref{eq22}.

It is also of interest to consider three spin measurements in directions $m$, $n$ and $r$. We then have that
\begin{align*}
S^m\circ (S^n\circ S^r)
  &=\left\{S_+^m\circ (S_+^n\circ S_+^r),S_+^m\circ (S_+^n\circ S_-^r), S_+^m\circ (S_-^n\circ S_+^r),\right.\\
  &\qquad S_+^m\circ (S_-^n\circ S_-^r),S_-^m\circ (S_+^n\circ S_+^r),S_-^m\circ (S_+^n\circ S_+^r),\\
  &\qquad\left.S_-^m\circ (S_+^n\circ S_-^r),S_-^m\circ (S_-^n\circ S_-^r)\right\}
\end{align*}
The first of these effects becomes
\begin{align*}
S_+^m\circ (S_+^n\circ S_+^r)&=S_+^m\circ\paren{\ab{\elbows{\phi _+^n,\phi _+^r}}^2S_+^n}
   =\ab{\elbows{\phi _+^n,\phi _+^r}}^2S_+^m\circ S_+^n\\
   &=\ab{\elbows{\phi _+^n,\phi _+^r}}^2\ab{\elbows{\phi _+^m,\phi _+^n}}^2S_+^m=c_{+++}S_+^m
\end{align*}
where we have defined $c_{+++}=\ab{\elbows{\phi _+^n,\phi _+^r}}^2\ab{\elbows{\phi _+^m,\phi _+^n}}^2$. The other effects and corresponding coefficients $c_{++-},c_{+-+},\ldots$, are similar. We conclude that
\begin{align*}
(S^n\circ S^r\mid S^m)&=\left\{c_{+++}S_+^m+c_{-++}S_-^m,c_{++-}S_+^m+c_{-+-}S_-^m,\right.\\
  &\qquad\left.c_{+-+}S_+^m+c_{--+}S_-^m,c_{+--}S_+^m+c_{---}S_-^m\right\}
  \intertext{and} 
  \paren{(S^r\mid S^n)\mid S^m}&=\left\{(c_{+++}+c_{+-+})S_+^m+(c_{-++}+c_{--+})S_-^m\right.\\
  &\qquad\left.(c_{++-}+c_{+--})S_+^m+(c_{-+-}+c_{---})S_-^m\right\}
\end{align*}
In general $(S^m\circ S^n)\circ S^r\ne S^m\circ (S^n\circ S^r)$ and we leave this to the reader.

\section{Conditioned Observables and Distributions}  
Let $A=\brac{A_x}$, $B=\brac{B_y}$ be observables on $H$ with value-spaces having cardinality $\ab{\Omega _A}=m$,
$\ab{\Omega _B}=n$. We say that $A$ and $B$ are \textit{complementary} if $(B_y\mid A_x)=\tfrac{1}{n}\,A_x$ and
$(A_x\mid B_y)=\tfrac{1}{m}\,B_y$ for all $x\in\Omega _A$, $y\in\Omega _B$. This condition says that when $A$ has a definite value $x$, then $B$ is completely random and vice versa. This is analogous to the complementary position and momentum observables of continuum quantum mechanics. We then have
\begin{align}                
\label{eq41}
A\circ B&=\brac{A_x\circ B_y\colon x\in\Omega _A,y\in\Omega _B}
    =\brac{\tfrac{1}{n}\,A_x,\ldots ,\tfrac{1}{n}\,A_x\colon x\in\Omega _A}\\
  \intertext{and}          
\label{eq42}       
B\circ A&=\brac{B_y\circ A_x\colon x\in\Omega _A,y\in\Omega _B}
     =\brac{\tfrac{1}{m}\,B_y,\ldots ,\tfrac{1}{m}\,B_y\colon y\in\Omega _B}
\end{align}
where there are $n$ terms $\tfrac{1}{n}\,A_x$ in \eqref{eq41} and $m$ terms in $\tfrac{1}{m}\,B_y$ in \eqref{eq42}. We also obtain
\begin{equation*}
(B\mid A)_y=\sum _{x\in\Omega _A}A_x\circ B_y=\sum _{x\in\Omega _A}\tfrac{1}{n}\,A_x=\tfrac{1}{n}\,I
\end{equation*}
for all $y\in\Omega _B$ and
\begin{equation*}
(A\mid B)_x=\sum _{y\in\Omega _B}B_y\circ A_x=\sum _{y\in\Omega _B}\tfrac{1}{m}\,B_y=\tfrac{1}{m}\,I
\end{equation*}
for all $x\in\Omega _A$. We conclude that $(B\mid A)$ and $(A\mid B)$ are identity observables. It is also interesting to note that
\begin{align*}
(A_x\circ B_y)\circ A_z&=\tfrac{1}{n}\,A_x\circ A_z, A_z\circ (A_x\circ B_y)=\tfrac{1}{n}\,A_z\circ A_x\\
(A_x\circ B_y)\circ B_z&=\tfrac{1}{n}\,A_x\circ B_z=\tfrac{1}{n^2}\,A_x\\
B_z\circ (A_x\circ B_y)&=\tfrac{1}{n}\,B_z\circ A_x=\tfrac{1}{nm}\,B_z
\end{align*}

Let $\brac{\psi _i}$, $\brac{\phi _i}$ be orthonormal bases for $H$ with $\dim H=d$, and let $A=\brac{P_{\psi _i}}$,
$B=\brac{P_{\phi _i}}$ be corresponding atomic observables. Since
\begin{equation*}
P_{\psi _i}\circ P_{\phi _j}=\ab{\elbows{\psi _i,\phi _j}}^2P_{\psi _j}
\end{equation*}
we see that $A$ and $B$ are complementary if and only if $\ab{\elbows{\phi _i,\psi _j}}^2=1/d$ for all $i,j=1,2,\ldots ,n$. Two orthonormal bases that satisfy this condition are called \textit{mutually unbiased} \cite{debz10,hz12,wf89}. There exist mutually unbiased bases in any finite-dimensional Hilbert space and such bases are important in quantum computation and information studies \cite{debz10,hz12,nc00}.

If $\rho\in\sscript (H)$ and $A\in\oscript (H)$, the \textit{distribution} of $A$ in the state $\rho$ is
\begin{equation*}
\Phi _A^\rho (x)=\rmtr (\rho A_x)
\end{equation*}
for all $x\in\Omega _A$. Then $\Phi _A^\rho$ defines a probability measures on $\Omega _A$ given by
\begin{equation*}
\Phi _A^\rho (X)=\sum _{x\in X}\Phi _A^\rho (x)
\end{equation*}
for all $X\in\Omega _A$. Clearly $\Phi _A^\rho$ is affine as a function of $\rho$. Also $\Phi _A^\rho$ is affine as a function of $A$ in the following sense. If $A_i\in\oscript (H)$ with the same value space $\Omega$ and $\lambda _i\ge 0$,
$i=1,2,\ldots ,n$ with $\sum\lambda _i=1$, then it is easy to verify that $\sum\lambda _iA_i\in\oscript (H)$ with value space
$\Omega$ where
\begin{equation*}
\sum\lambda _iA_i=\brac{\sum\lambda _iA_{ix}\colon x\in\Omega}
\end{equation*}
We then have that
\begin{equation*}
\Phi _{\sum\lambda _iA_i}^\rho (x)=\rmtr\paren{\rho\sum\lambda _iA_{ix}}=\sum\lambda _i\rmtr (\rho A_{ix})
   =\sum\lambda _i\Phi _{A_i}^\rho (x)
\end{equation*}
for all $x\in\Omega$. Hence,
\begin{equation*}
\Phi _{\sum\lambda _iA_i}^\rho =\sum\lambda _i\Phi _{A_i}^\rho
\end{equation*}

For the sequential product $A\circ B$ we have that
\begin{equation*}
\Phi _{A\circ B}^\rho (x,y)=\rmtr (\rho A_x\circ B_y)=\rmtr (A_x^{1/2}\rho A_x^{1/2}B_y)=\rmtr\sqbrac{(\rho\mid A_x)B_y}
\end{equation*}
The \textit{left marginal} of $\Phi _{A\circ B}^\rho$ is defined by
\begin{equation*}
L\Phi _{A\circ B}^\rho (x)=\sum _{y\in\Omega _B}\Phi _{A\circ B}^\rho (x,y)=\rmtr (\rho A_x)=\Phi _A^\rho (x)
\end{equation*}
for all $x\in\Omega _A$ so that $L\Phi _{A\circ B}^\rho =\Phi _A^\rho$. More interestingly, the \text{right marginal} of
$\Phi _{A\circ B}^\rho$ becomes
\begin{equation*}
R\Phi _{A\circ B}^\rho (y)=\sum _{x\in\Omega _A}\Phi _{A\circ B}^\rho (x,y)=\rmtr\sqbrac{\rho (B\mid A)_y}
   =\Phi _{(B\mid A)}^\rho (y)
\end{equation*}
for all $y\in\Omega _B$ so that $R\Phi _{A\circ B}^\rho =\Phi _{(B\mid A)}^\rho$. For the conditional observable $(B\mid A)$ we have that
\begin{align*}
\Phi _{(B\mid A)}^\rho (y)&=\rmtr\paren{\rho\sum _{x\in\Omega _A}A_x\circ B_Y}
   =\rmtr\paren{\sum _{x\in\Omega _A}A_x^{1/2}\rho A_x^{1/2}B_y}\\
   &=\rmtr\sqbrac{(\rho\mid A)B_y}=\Phi _B^{(\rho\mid A)}(y)
\end{align*}
for every $y\in\Omega _B$. Hence, $R\Phi _{A\circ B}^\rho=\Phi _{(B\mid A)}^\rho =\Phi _B^{(\rho\mid A)}$.

We now consider some special cases. If $A$ and $B$ are complementary, we have that
\begin{equation*}
\Phi _{A\circ B}^\rho (x,y)=\tfrac{1}{n}\,\rmtr (\rho A_x)=\tfrac{1}{n}\,\Phi _A^\rho (x)
\end{equation*}
for all $(x,y)\in\Omega _A\times\Omega _B$. Moreover,
\begin{equation*}
\Phi _{(B\mid A)}^\rho (y)=\rmtr\paren{\rho\sum _{x\in\Omega _A}\tfrac{1}{n}\,A_x}=\tfrac{1}{n}
\end{equation*}
Thus, $\Phi _{(B\mid A)}^\rho$ is completely random.

As another example, let $A=\brac{P_{\phi _x}\colon x\in\Omega _A}$ be atomic and let $B\in\oscript (H)$ be arbitrary. Then
$\Phi _A^\rho (x)=\elbows{\phi _x,\rho\phi _x}$ and
\begin{align*}
\Phi _{A\circ B}^\rho (x,y)&=\rmtr (P_{\phi _x}\rho P_{\phi _x}B_y)
   =\elbows{\phi _x,\rho\phi _x}\rmtr\paren{\ket{\phi _x}\bra{\phi _x}B_y}\\
   &=\elbows{\phi _x,\rho\phi _x}\elbows{\phi _x,B_y\phi _x}=\Phi _A^\rho (x)\Phi _B^{P_{\phi _x}}(y)
\end{align*}
We also have that
\begin{equation*}
\Phi _{(B\mid A)}^\rho (y)=\rmtr\paren{\sum _{x\in\Omega _A}P_{\phi _x}\rho P_{\phi _x}B_y}
  =\sum _{x\in\Omega _A}\Phi _A^\rho (x)\Phi _B^{P_{\phi _x}}(y)
\end{equation*}

\section{Defining Joint Probabilities}  
\centerline{``A good definition is worth a hundred theorems.'' --Unknown}
\bigskip
\noindent In classical probability theory, events are represented by sets and if $A$ and $B$ are events and $\mu$ is a probability measure, then $\pscript _\mu (A\hbox{ and }B)=\mu (A\cap B)$ is their joint probability. This definition is not adequate for quantum mechanics. One reason for this is that it does not take account of which event is observed first. If such a temporal order is considered, then the first measurement may interfere with the second, resulting in quantum interference. Another problem is caused by the joint additivity of $\pscript _\mu$ which we shall discuss later.

If $A,B\in\oscript (H)$, $\rho\in\sscript (H)$, then the standard definition of the joint probability is
\begin{equation}                
\label{eq51}
\pscript _\rho (A_X\hbox{ then }B_Y)=\rmtr\sqbrac{\iscript _X^A(\rho )B_Y}
\end{equation}
where $\iscript ^A$ is an $A$-compatible instrument \cite{bcl95,bgl95,hz12}. We interpret this as meaning that if the system is initially in the state $\rho$ and $A$ is measured first giving a value in $X$ and next $B$ is measured giving a value in $Y$, then their joint probability is the right side of \eqref{eq51}. We believe that \eqref{eq51} is not a satisfactory definition because this joint probability should depend on $A_X$ and not on an instrument measuring $A$. In particular, if $\iscript ^A$ is a trivial
$A$-compatible instrument $\iscript _X^A(\rho )=\rmtr (\rho A_X)\eta$, then
\begin{align*}
\rmtr\sqbrac{\iscript _X^A(\rho )B_Y}&=\rmtr\sqbrac{\rmtr (\rho A_X)\eta B_Y}=\rmtr (\rho A_X)\rmtr (\eta B_Y)\\
   &=\pscript _\rho (A_X)\pscript _\eta (B_Y)
\end{align*}
Hence, $\pscript _\rho (A_X\hbox{ then }B_Y)$ can be any number less than or equal to $\rmtr (\rho A_X)$ depending on the choice of $\eta$. Also, the conditional output state \cite{hz12,oza85} becomes
\begin{equation*}
\rhotilde _X=\frac{1}{\rmtr\sqbrac{\iscript _X^A(\rho )}}\,\iscript _X^A(\rho )=\eta
\end{equation*}
and this has nothing to do with $A$ which is again unsatisfactory. Moreover, if we include an instrument measuring $A$, why not also include an instrument measuring $B$? We would then define
\begin{equation*}
\pscript _\rho (A_X\hbox{ then }B_Y)=\rmtr\sqbrac{\iscript _Y^B\paren{\iscript _X^A(\rho )}}
\end{equation*}
which gives different results than \eqref{eq51}.

Instead of an arbitrary $A$-compatible instrument, we suggest employing the unique $A$-compatible L\"uders instrument
$\lscript ^A$. This overcomes the previously discussed problems. Moreover, $\lscript ^A$ is the canonical $A$-compatible instrument because any $A$-compatible instrument has the form $\iscript _x=\escript _x\circ\lscript _x^A$ where
$\brac{\escript _x\colon x\in\Omega _A}$ is a set of channels \cite{hz12,lud51}. With this assumption \eqref{eq51} becomes
\begin{align}                
\label{eq52}
\pscript _\rho (A_X\hbox{ then }B_Y)&=\rmtr\sqbrac{\lscript _X^A(\rho )B_Y}
    =\rmtr\paren{\sum _{x\in X}A_x^{1/2}\rho A_x^{1/2}B_Y}\notag\\
    &=\sum _{x\in X}\rmtr (\rho A_x^{1/2}B_YA_x^{1/2})=\sum _{x\in X}\rmtr (\rho A_x\circ B_Y)\\
    &=\sum _{x\in X}\pscript _\rho (A_x\circ B_Y)\notag
\end{align}
Moreover, the conditional output state becomes
\begin{equation*}
\rhotilde _X=\frac{1}{\rmtr\sqbrac{\lscript _X^A(\rho )}}\,\lscript _X^A(\rho )
   =\frac{1}{\rmtr (\rho A_X)}\sum _{x\in X}(A_x\circ\rho )
\end{equation*}

Even this last definition of $\pscript _\rho (A_X\hbox{ then }B_Y)$ is not satisfactory. This is because \eqref{eq51} and \eqref{eq52} are additive in the first variable. That is
\begin{equation}                
\label{eq53}
\pscript _\rho (A_{\cup X_i}\hbox{ then }B_Y)=\sum _i\pscript _\rho (A_{X_i}\hbox{ then }B_Y)
\end{equation}
whenever $X_i\cap X_j=\emptyset$ for $i\ne j$. Now a measurement of $A$ can interfere with a later measurement of $B$ so one should not expect \eqref{eq53} to hold. In fact, \eqref{eq53} is the defining property of classical probability theory in which events do not interfere. Of course, \eqref{eq51} and \eqref{eq52} are also additive in the second variable, but this is not a problem because the measurements of $B$ is after the measurement of $A$ so there is no interference.

We believe that the natural and correct definition of the joint probability is:
\begin{equation}                
\label{eq54}
\pscript _\rho (A_X\hbox{ then }B_Y)=\pscript _\rho (B_Y\mid A_X)=\rmtr (\rho A_X\circ B_Y)
\end{equation}
This is just the definition that we have used in the previous sections of this article. Notice that the difference between \eqref{eq54} and the last expression in \eqref{eq52} is the lack of additivity in \eqref{eq54}.

In order to investigate additivity more closely, we make the following definition. For $a,b,c\in\escript (H)$ with $a\perp b$, we say that $a$ and $b$ are \textit{additive relative to} $c$ if 
\begin{equation}                
\label{eq55}
(a+b)\circ c=a\circ c+b\circ c
\end{equation}
and when \eqref{eq55} holds we write $(a,b\colon c)$. We can rewrite \eqref{eq55} as
\begin{equation}                
\label{eq56}
(a+b)^{1/2}c(a+b)^{1/2}=a^{1/2}ca^{1/2}+b^{1/2}cb^{1/2}
\end{equation}
We show in the next lemma that $(a,b\colon c)$ is a weakening of the compatibility of $a,c$ and $b,c$. This makes sense because compatibility corresponds physically to noninterference which, as mentioned previously, is related to additivity.

\begin{lem}    
\label{lem51}
If $a,c$ an $b,c$ are compatible, the $(a,b\colon c)$.
\end{lem}
\begin{proof}
If $ac=ca$ and $bc=cb$, then $a,c$ and $b,c$ can be simultaneously diagonalized. Hence $a^{1/2}c=ca^{1/2}$ and
$b^{1/2}c=cb^{1/2}$. Also, $c$ and $a+b$ can be simultaneously diagonalized so $(a+b)^{1/2}c=c(a+b)^{1/2}$. Therefore,
\begin{equation*}
(a+b)^{1/2}c(a+b)^{1/2}=c(a+b)=ca+cb=a^{1/2}ca^{1/2}+b^{1/2}cb^{1/2}
\end{equation*}
so $(a,b\colon c)$.
\end{proof}

The next example shows that the converse of Lemma~\ref{lem51} does not hold. Thus, there are noncompatible pairs that are still additive.

\begin{exam}       
Let $a,b,c\in\escript (\complex ^3)$ be the following effects
\begin{equation*}
a=\begin{bmatrix}1&0&0\\0&0&0\\0&0&0\\\end{bmatrix},\quad b=\begin{bmatrix}0&0&0\\0&0&0\\0&0&1\\\end{bmatrix},
   \quad c=\tfrac{1}{2}\begin{bmatrix}1&1&0\\1&1&0\\0&0&0\\\end{bmatrix}
\end{equation*}
Then $cb=0$ and $ac\ne ca$ so $a$ and $c$ are not compatible. However,
\begin{equation*}
(a+b)\circ c=a\circ c+b\circ c=\tfrac{1}{2}\,a\rlap{$\qquad\qquad \Box$}
\end{equation*}
\end{exam}

The next lemma characterizes additivity for sharp $a$ and $b$.

\begin{lem}    
\label{lem52}
If $a,b\in\escript (H)$ are sharp with $a\perp b$, then $(a,b\colon c)$ if and only if $acb=0$.
\end{lem}
\begin{proof}
Since $a$ and $b$ are sharp and $a\perp b$, we have that $ab=0$ \cite{hz12} so $a+b$ is sharp. Hence,
\begin{align*}
(a+b)\circ c&=(a+b)c(a+b)=aca+bcb+acb+bca\\
  &=a\circ c+b\circ c+acb+bca
\end{align*}
Therefore, $(a,b\colon c)$ if and only if $acb+bca=0$ which is equivalent to\newline $acb=0$.
\end{proof}

We do not know a generalization of this lemma for unsharp $a,b\in\escript (H)$.

\end{document}